\documentclass[conference]{IEEEtran}

\usepackage{amsmath,amsthm,amssymb}
\usepackage{graphics, subfigure}
\usepackage{cite}
\usepackage{tikz}
\usepackage{mathtools}
\usetikzlibrary{shapes,arrows}
\usepackage{algorithm, algorithmic}
\usepackage[outdir=./]{epstopdf}
\usepackage{color}
\usepackage{multirow}
\usetikzlibrary{automata,positioning}
\usepackage{wrapfig}
\usepackage{rotating}

\usepackage{bbm}
\usepackage{url}
\usepackage{bbm}

\usepackage{enumitem}

\newtheorem{assumption}{Assumption}
\newtheorem{assum}{Assumption}
\newtheorem{theorem}{Theorem}
\newtheorem{lemma}{Lemma}

\newtheorem{definition}{Definition}
\newtheorem{remark}{Remark}
\newtheorem{cor}{Corollary}

\newtheorem{lem}{Lemma}
\newtheorem{example}{Example}
\title{Competition among Ride Service Providers with Autonomous Vehicles}

\begin{document}
\author{
\IEEEauthorblockN{Arnob Ghosh}
\IEEEauthorblockA{Electrical and Electronic Engineering Department\\ Imperial College of London\\ arnob.ghosh@imperial.ac.uk }
\and 
\IEEEauthorblockN{Randall Berry}
\IEEEauthorblockA{ Electrical and Computer Engineering Department\\  Northwestern University\\ rberry@ece.northwestern.edu}
}
\maketitle
\newcommand{\G}{\mathcal{G}}
\newcommand{\V}{\mathcal{V}}
\newcommand{\E}{\mathcal{E}}
\begin{abstract}
Autonomous vehicles (AVs) are attractive for ride service providers (RSPs) in part because they 
eliminate the need to compete for human drivers. We investigate a scenario where two RSPs with AVs  compete for customers. We model the problem as a game  where the RSPs select prices for each origin-destination pair over multiple time periods in an underlying graph representing the customers' desired trips.  Each RSP also decides the number of AVs to be stationed at each node at each time period to serve the customers' demands. The number of customers who avail service of a RSP depends on the price selected by the RSP and its competitor. Since the strategy choices available to a RSP depends on its competitor, we seek to compute a Generalized Nash equilibrium (GNE). We show that there may be multiple GNEs. However, when a RSP selects prices in order to deter its competitor when it is not serving  a source-destination pair, the game has a potential function and admits a unique GNE. We also compare the competitive prices with a monopoly price where only one RSP is in the market. Numerically, we show that if a network consists of two equal size spatial clusters of demand where the demand between clusters is low, the RSPs may partition the market, i.e, one cluster is served by only one RSP. Hence, the competitive price may become close to the monopoly price. %demand is not balanced across the network  and the demand is very high, the competitive price becomes equal to the monopoly price where only % A lot of research has been devoted towards finding an optimal price for a RSP. However, 
\end{abstract}
\section{Introduction}
The transportation system is going through disruptive changes. Ride service providers (RSPs) (e.g., Uber, and Lyft) are increasingly serving passengers in most cities across the world \cite{ad1,ad2}. RSPs are exploring using autonomous vehicles (AVs), in part because  in traditional human driven vehicles, they need to compete with other RSPs for drivers and compensate them for their time spent driving. Instead, they can reduce operation costs by owning the fleet of AVs. The RSPs can control and dispatch AVs in order to meet the customers' demands. 

If RSPs own AVs, they will no longer compete for drivers. This may result in RSPs lowering prices to compete for customers. Whether customers see this benefit will depend in turn on how competitive the RSP market is and if RSPs, for example, have an incentive to segment the market geographically to reduce competition. We seek to gain insight into this question. 

We consider a stylized model where two RSPs operate over a region consisting of multiple locations. We formulate the region as a graph with the locations as the nodes. We assume that each RSP dispatches vehicles in order to satisfy the estimated demand at different locations. The demand seen by a certain RSP for an origin-destination pair $(j,l)$ depends on the price selected by the RSP, its competitor, and the total number of customers willing to go from origin $j$ to destination $l$ at each time. Each RSP decides the prices and the number of available vehicles it would dispatch for each origin-destination pair at each time instance in order to maximize the total profit over a given time horizon. We formulate the competition among the RSPs as a game theoretic problem and  show that the strategy space of a RSP is also constrained by the strategy of other RSP.  We seek to obtain a Generalized Nash equilibrium (GNE) \cite{gne}. We show that  there can be multiple GNEs depending on the action of the competitor when it is not serving an origin-destination pair. We show that when a RSP ``tries to deter" its competitor maximally, the resulting game has a potential function and such a GNE can be obtained by solving a convex optimization problem. Other GNEs appear to be more difficult to obtain. 

In order to investigate whether competition would exist among the RSPs, we compare the outcome in the duopoly scenario with another scenario where there is only one (monopoly) RSP. The RSP selects a price for each origin-destination pair ($j,l)$. The demand is then realized as a function of the price and the number of customers willing to travel from location $j$ to location $l$. We formulate a convex optimization problem where the monopoly RSP decides prices and the number of free vehicles to dispatch for each origin-destination pair in order to maximize its revenue. We denote the price selected by the RSP as the {\em monopoly price}. We characterize when the equilibrium prices selected in the duopoly model and the monopoly price for each origin-destination pair become identical. 

%\begin{color}{blue}
In Section~\ref{sec:numerical}, we numerically investigate a network consisting of two clusters of demand. The demand from one cluster to another one is varied. We show that when the demand is not balanced and the capacities of the RSPs are small, RSPs may partition the market between them. A cluster is served by only one RSP. {\em Hence, the price becomes equal to the monopoly price even thought there are two RSPs. However, as the demand becomes more balanced, the price decreases.} The price also decreases when the capacities of each RSP becomes high. 
%\end{color}

{\em Related Literature}: Characterizing optimal price mechanisms for a RSP is of  great interest. Our work is related to several works on platform price design for a RSP \cite{feng,banerjee,sun,taylor,gurvich}. These papers model the stochasticity of the driver's opportunity cost and the demand, and determine optimal price mechanisms for a ride-sharing platform. However, the above papers did not consider the spatial and temporal structure of the underlying demand. Further, the above papers did not model the competition among multiple RSPs. In \cite{bernstein}, a competition model among RSPs was considered.  The strategic interaction between a regulator and an RSP was studied in \cite{yu2020}. However, \cite{bernstein,yu2020}  did not consider the spatial variation of the demand.

Recently, \cite{bimpikis,chen2019,berry} investigated the spatial variation of prices across a network for a RSP. However, these papers did not consider a competition model as there is only one RSP. Further, \cite{bimpikis,chen2019} did not consider any temporal variation of demand at a location. Thus, they also did not consider the time vehicles take going from one location to another, and so they did not model the routing decision of vehicles and its potential impact on the pricing. In \cite{berry}, the travel time was considered, however, the temporal variation of demand was not modeled. In \cite{YANG}, the optimal supply pattern for a taxi service provider over a network was considered for a fixed price set by the government. However, the dynamic nature of the price and the impact of competition among multiple RSPs on the pricing as well as the supply behavior was not considered.  

In other related literature, competition in a network model is widely studied by employing the Cournot network game \cite{bimpikis2019,barquin,cournot1}. However, in a transportation network, one needs to consider the routing decision of vehicles based on the origin and destination, which is not modeled in the Cournot game. The Cournot model assumes that the price is determined by the total supply and the price is same for the firms participating in the market at the same location. In contrast, we consider a competition model among RSPs where the prices can vary across the origin-destination pairs and may be different for different RSPs even for the same origin-destination pair. 

\section{System Model}
We assume that two RSPs are competing for customers in a region consisting of several locations. The region is modeled as a graph $\G=\{\V,\E\}$. For example, if the region is a city, then the city can be divided into a grid graph where every node represents each square block of the city. The edges are directed. If an edge exists from node $i$ to node $j$, then traffic can flow from node $i$ to node $j$.   We assume that the graph is strongly connected. 

%\begin{color}{red}
Each RSP $i$ has a given number of AVs. A RSP can dispatch vehicles to several locations in order to satisfy the demand. The overall time horizon is $T$ and time is slotted with each slot  of duration $\tau$. The specific duration of a slot can be of any fixed value. The demand for travel to location $l$ from location $j$ is denoted as $D_{j,l,t}$ which is estimated by the RSPs. We assume that the estimation is correct throughout the $T$ time periods.  In Section~\ref{sec:stoc_demand} we extend our analysis to the setting where the RSPs only know the distribution of the demand within a time window.  %\end{color}

RSP $i$ selects a price $p_{i,j,l,t}$ at location $j$ towards location $l$ at time $t$. There exists an upper limit $p_{max}$ on the price. If a price exceeds the upper limit, no customer would choose the RSP. The customers would choose one of the other RSPs or not to ride according to the demand model described next.

\subsection{Demand Model}

The demand of customers for each RSP $i$ at location $j$ towards destination $l$  is assumed to be the following
\begin{align}\label{eq:dem}
D_{i,j,l,t}=\max\{D_{j,l,t}(\dfrac{1}{2}-\dfrac{1}{p_{max}}p_{i,j,l,t}+\dfrac{1}{2p_{max}}p_{k,j,l,t}),0\}
\end{align}
where $p_{k,j,l,t}$ is the price set by the RSP $k$, $k\neq i$. The demand for RSP $i$ would increase if $p_{k,j,l,t}$ increases since more customers would prefer RSP $i$ compared to RSP $k$. On the other hand, as the price $p_{i,j,l,t}$ increases the demand for RSP $i$ would decrease. Note that, compared to the basic Bertrand model for price competition, all the customers here do not necessarily select the RSP with the lowest price. This captures other preferences that customers may have regarding a given RSP that can arise in a practical market.  Such a demand model has been widely considered in the literature, e.g.  \cite{niyato2008,el2003,garmani2019}. 

%\begin{color}{red}
The demand for RSP $i$ is impacted more by the change in the price $p_{i,j,l,t}$ compared to the price of RSP $k$, $p_{k,j,l,t}$. Hence, the co-efficient corresponding to $p_{i,j,l,t}$ is higher compared to the co-efficient corresponding to $p_{k,j,l,t}$ in the demand expression. For the ease of exposition, we consider that the co-efficient corresponding to $p_{i,j,l,t}$ is exactly twice that of the co-efficient corresponding to $p_{k,j,l,t}$. Our analysis will go through for any ratio of the co-efficient values greater than $1$. %\end{color}

At any price greater than or equal to $p_{max}$ customers will not avail any service. Thus, $p_{max}$ is the least upper bound on the prices at which customers are willing to avail a service. When $p_{k,j,l,t}=p_{max}$, this is equivalent to a scenario where RSP $k$ is not present in the market from source $j$ to destination $l$ during time interval $t$. Thus, RSP $i$ can enjoy a monopoly power and its demand $D_{i,j,l,t}$ will only be zero when $p_{i,j,l,t}=p_{max}$. In general, note that $D_{i,j,l,t}$ is zero when
\begin{align}\label{eq:pricez}
p_{i,j,l,t}\geq \dfrac{p_{max}}{2}+\dfrac{p_{k,j,l,t}}{2}.
\end{align}
Hence,  if $p_{k,j,l,t}<p_{max}$, the demand $D_{i,j,l,t}$ can be $0$  even when $p_{i,j,l,t}<p_{max}$ since the other RSP is more attractive to the customers. 

If both $D_{i,j,l,t}>0$ and $D_{k,j,l,t}>0$, then the total number of customers who avail either of the RSPs' services is given by
\begin{align}\label{eq:totdem}
D_{j,l,t}(1-\dfrac{1}{2p_{max}}p_{i,j,l,t}-\dfrac{1}{2p_{max}}p_{k,j,l,t}).
\end{align}
If both $p_{i,j,l,t}$ and $p_{k,j,l,t}$ are equal to $0$, the total demand reaches the maximum value, $D_{j,l,t}$. Note that when the demand is lower than $D_{j,l,t}$, then the remaining customers would take another mode of transportation to reach the destination or may choose not to travel at all.

%{\bf Doesn't this contradict assumption 1?} 
%We also assume that if a RSP does not serve customers for a given source-destination pair $(j,l)$, then the price for that route would be $p_{max}$  to reflect the fact that the demand for the RSP from node $j$ to node $l$ would be $0$. %there is not enough number of vehicles available at a location by a RSP, then the price would become $p_{max}$ at a certain time in order to reflect on the fact that the RSP can not fulfill the total demand of the customers. 

\subsection{Pricing strategy of the RSPs}
Since the RSPs only have AVs, they do not need to select any prices for human drivers. However, they need to select prices for the customers who are availing the services. 
Each RSP decides the prices for each source-destination pair at time $t$ in order to maximize its revenue. The revenue of RSP $i$ is
\begin{align}
\sum_{t=1}^{T}\sum_{j}\sum_{l}p_{i,j,l,t}D_{i,j,l,t}.
\end{align}
Note that the demand for a RSP depends on the price selected by the other RSP (cf.~(\ref{eq:dem})). 

%\begin{color}{red}
Note that in practice, the RSPs may optimize for a smaller time window $T$ since the estimate for demand may be accurate for a smaller time-scale and can employ a sliding window approach by optimizing first from $t=1,\ldots, T$, then $t=2,\ldots, T+1$ and so on.%\end{color}

\subsection{Routing of the vehicles by the RSPs}
Each RSP also needs to route its vehicles in order to meet the demand. Note that a RSP can only serve customers at a certain location at a given time if it has vehicles at that location at that time. If a vehicle is not serving any customer, then, the vehicle is a {\it free vehicle} and can be routed by the RSP to some other location for picking up customers. Let $u_{i,j,l,t}$ be the number of free vehicles which have been routed  by RSP $i$ to location $l$ from location $j$ at time $t$. Let $x_{i,j,t}$ be the total number of vehicles operated by RSP $i$ at time $t$  at location $j$. 

A RSP can indirectly control the number of customers it would serve to a destination at a certain time. For example, if the demand at certain location is very small, the RSP may select a higher price for that location which would drive down the demand to that location so that it may instead deploy vehicles to other locations. 

A vehicle takes a certain time to move from one location to another. We define the tensor $A^t$ that specifies the exact number of time slots a vehicle takes  to move from location $i$ to location $j$ at time slot $t$. For example, if it takes $2$ time slots to go from location $i$ to location $j$ at time $t$, then in the tensor $A^t$, the element $A^t_{i,j,2}$ is $1$. We assume that at the start of the horizon each RSP knows all of the elements of $A^t$ for all time $t$ within the horizon.  Note that if the time horizon $T$ is small, this is a reasonable assumption, as the RSPs can have accurate predictions regarding the time a vehicle takes to go from one node to another node. 

Formally,
\begin{align}
A^t_{j,l,\tau}=\begin{cases} 1,\quad \text{if it takes exactly $\tau$ slots}\nonumber\\\text{ \indent to reach to location $l$ from $j$, }\nonumber\\
0, \quad \text{otherwise}.\end{cases}
\end{align}
We assume that $A^t_{j,l,0}=0$ for all source-destination pairs. Thus, vehicles take at least one time slot to travel over link $(j,l)$. Note that a vehicle may take different times to reach node $l$ from node $j$ and to reach node $j$ from node $l$;  hence, $A^t_{j,l,\tau}$ may not be symmetric. 
For all customers who have chosen RSP $i$ at time $t$ to travel from $j$ towards $l$, those vehicles will again be available after $\tau$ time slots at location $l$ if $A^t_{j,l,\tau}=1$. Thus, the number of vehicles owned by RSP $i$ that reach  location $j$ at time $t$ because of the rides from other locations towards location $j$ is given by
\begin{align}
\sum_{\tau<t}\sum_{l}D_{i,l,j,\tau}A^{\tau}_{l,j,t-\tau}.
\end{align}
Similarly, the number of free vehicles that will reach location $j$ when free vehicles are routed by the RSP $i$ from other locations to location $j$ is given by
\begin{align}
\sum_{\tau<t}\sum_{l}u_{i,l,j,\tau}A^{\tau}_{l,j,t-\tau}.
\end{align}

Therefore, the number of available vehicles of RSP $i$ at the end of time $t$ is given by
\begin{eqnarray}
x_{i,j,t}=x_{i,j,t-1}+\sum_{\tau<t}\sum_{l}D_{i,l,j,\tau}A^{\tau}_{l,j,t-\tau}+\nonumber\\
\sum_{\tau<t}\sum_{l}u_{i,l,j,\tau}A^{\tau}_{l,j,t-\tau}-\sum_{l}D_{i,j,l,t}-\sum_{l}u_{i,j,l,t}.\label{eq:state}
\end{eqnarray}
The last two terms represent the number of rides from location $j$ to other locations at time $t$ and the number of free vehicles routed to other locations from location $j$ at time $t$. We assume that $x_{i,j,0}$ is given and known to all the RSPs. Hence, the initial locations of the vehicles owned by RSP $i$ are known. Note that  since we have assumed that $A^{\tau}_{l,j,0}=0$, for any origin-destination pair $(l,j)$, thus, $D_{j,j,\tau}A^{\tau}_{j,j,t-\tau}=0$ for $t=\tau$. 

Furthermore, we assume that the vehicle supplies must satisfy the following two constraints:
\begin{eqnarray}
x_{i,j,t}\geq 0,\quad u_{i,j,l,t}\geq 0.\label{eq:control}
\end{eqnarray}
%The constraint in (\ref{eq:capacity}) represents the capacity constraint as the sum of all the available vehicles at all the locations can not exceed the capacity of RSP $i$, $C_i$.

The constraints in (\ref{eq:control}) represent that the number of free vehicles  which are routed and the number of available vehicles at any location can not be negative. The constraints ensure that  the demand can not exceed the total number of available vehicles. Note that since the initial locations of the vehicles are given, the constraint in (\ref{eq:control}) will ensure that the total number of vehicles (sum of the number of vehicles which are in transit and which are free) never exceed the capacity $C_i$ of RSP $i$. 

Finally, we require the price $p_{i,j,l,t}$ must be non-negative and less than or equal to $p_{max}$:
\begin{align}\label{eq:priceconstrain}
0\leq p_{i,j,l,t}\leq p_{max}.
\end{align}

\subsection{Strategy of Each RSP}
We formulate the RSP's price selection and routing decision as a game theoretic problem where each RSP takes its decision in order to maximize its own profit. Hence, RSP $i$ solves the following
\begin{eqnarray}
\mathcal{P}: \text{maximize }&  F_i=(\sum_{t}\sum_{j}\sum_{l}p_{i,j,l,t}D_{i,j,l,t}-c_{j,l,t}D_{i,j,l,t}\nonumber\\& -c^{\prime}_{j,l,t}u_{i,j,l,t})\nonumber\\
\text{subject to }&  (\ref{eq:state}), (\ref{eq:control}), (\ref{eq:priceconstrain}).\nonumber
\end{eqnarray}
Here, $c_{j,l,t}$ represent the cost of routing one vehicle from $j$ to $i$ at time $t$, so that the second term in the objective represents the cost for all trips from location $j$ to location $l$. Likewise, $c'_{j,l,t}$ represents the per-vehicle cost of routing a free vehicle from $j$ to $i$ at time $t$, which in general we allow to be different from $c_{i,j,t}$, e.g., if it requires less energy to route a free vehicle. Hence, the third term in the objective represents the total cost for rerouting free vehicles from one location to other location. The decision variables are $p_{i,j,l,t}$ and $u_{i,j,l,t}$. The parameters $D_{i,j,l,t}$ and $x_{i,j,t}$ depend implicitly on the two decision variables. 

We assume that vehicles move from location $j$ to location $l$ by taking the path with the shortest time. For example, the vehicles may be equipped with an app like Google Maps that provides this information. The decision  of how to optimally choose paths for each route is beyond the scope of this paper. 

\begin{definition}\label{defn:strategy}
For $i=1,2$, let $S_i=\{p_{i,j,l,t},u_{i,j,l,t}\}$  be the strategy of RSP $i$ and let $\mathcal{S}_i$ be the strategy space of the RSP. Let $S_{-i}$ be the strategy of the RSP other than $i$ and let $\mathcal{S}_{-i}$ be the strategy space of the RSP other than $i$. 
\end{definition}
Note from (\ref{eq:dem}) that the demand $D_{i,j,l,t}$ inherently depends on the strategy of the other RSP ($p_{k,j,l,t}$). Specifically, given the price $p_{k,j,l,t}$, $k\neq i$, the price $p_{i,j,l,t}$ completely specifies the demand $D_{i,j,l,t}$.

The strategy space $\mathcal{S}_i$ for $i=1,2$ is characterized by the set of constraints of the optimization problem $\mathcal{P}$. Note from (\ref{eq:state}) that the strategy space $\mathcal{S}_i$ depends on the strategy $S_{-i}$ since the demand $D_{i,j,l,t}$ depends on the price selected by the RSP $k\neq i$. At times to make this explicit, we will denote RSP $i$'s strategy space as $\mathcal{S}_i(S_{-i})$.
When the strategy space of a player also depends on the strategy of other player, the appropriate solution concept is a generalized Nash equilibrium (GNE) instead of a Nash equilibrium. In the following, we define the strategy profile which constitutes a GNE.

\begin{definition}\label{degn:gne}
A Generalized Nash equilibrium (GNE) strategy profile for the RSPs is the strategy vector $\{S_1,S_2\}$ such that for each RSP $i=1,2$
\begin{align}
F_i(S^{\prime}_i,S_{-i})\leq F_i(S_i,S_{-i}) \forall S^{\prime}_i\in \mathcal{S}_i(S_{-i}).
\end{align}
where $F_i(\cdot)$ is the objective of RSP $i$ of the optimization problem $\mathcal{P}$. %and $\mathcal{S}_i$ is defined by the strategy $S_{-i}$.
\end{definition}

\subsection{Multiple GNEs}
Note from (\ref{eq:pricez}) that if $D_{i,j,l,t}=0$, there can be multiple values of $p_{i,j,l,t}$ for which $D_{i,j,l,t}=0$.   Hence, different strategies can give the same demand for a RSP. Note that the price selected by RSP $i$, $i=1,2$ when $D_{i,j,l,t}=0$ can impact the strategy space $\mathcal{S}_{-i}$. Thus, there may be multiple GNEs. 

\begin{assumption}\label{assum1}
We consider that when $D_{i,j,l,t}=0$ for $i=1,2$ the price $p_{i,j,l,t}$ satisfies
\begin{align}\label{eq:wgne}
p_{i,j,l,t}=\dfrac{p_{max}}{2}+\dfrac{p_{k,j,l,t}}{2}
\end{align}
for $k\in \{1,2\}, k\neq i$. 
\end{assumption}
 In an equilibrium when $D_{i,j,l,t}=0$, any price $p_{i,j,l,t}$ as in (\ref{eq:pricez}) would not influence the revenue of RSP $i$. However, it would influence the strategy $S_{-i}$. Specifically, when the price is set as in (\ref{eq:wgne}), it is the smallest price the RSP $i$ can set when $D_{i,j,l,t}=0$. Thus, RSP $k\in \{1,2\}, k\neq i$, also needs to lower its price in order to maintain the same demand $D_{k,j,l,t}$. This corresponds to ``maximum way''  RSP $i$ can deter RSP $k$ or a `credible threat' from RSP $i$.  In a non-cooperative competitive setting, a RSP may want to inflict maximum harm on its competition. Assumption 1 represents the above behavior of a RSP. 
 
 Note that if $D_{i,j,l,t}>0$ for $i=1,2$ and for all $(j,l)$, then the prices of the RSPs becomes unique. We later show that under Assumption~\ref{assum1}, by solving a convex optimization problem we obtain a GNE. However, for other price choices when $D_{i,j,l,t}=0$, the strategy space becomes non-convex. Thus, determining a GNE may become more computationally challenging. 

%We, thus, define the worst GNE (WGNE) and the best GNE (BGNE) corresponding to the equilibrium revenue for RSP $i$, when $D_{i,j,l,t}=0$. 
%\begin{definition}
%The NE strategy profile $(S_1,S_2)$ is the WGNE if for all $i=1,2$,
%\begin{align}
%p_{i,j,l,t}=\dfrac{p_{max}}{2}+\dfrac{p_{k,j,l,t}}{2}
%\end{align}
%for $k\in \{1,2\}, k\neq i$, when $D_{i,j,l,t}=0$. 
%
%The NE strategy profile $(S_1,S_2)$ is the BGNE if for all $i=1,2$,
%$p_{i,j,l,t}=p_{max}$ when $D_{i,j,l,t}=0$. 
%\end{definition}

\subsection{Convex Optimization}
Note from (\ref{eq:dem}) that the objective in Problem $\mathcal P$ is not concave due to the in piecewise linear structure of $D_{i,j,l,t}$.    In the following, we give an equivalent representation which is a convex optimization problem. First, we introduce a notation
\begin{definition}
Let $\tilde{F}_i(S_i,S_{-i})$ be
\begin{align}
&\tilde{F}_i(S_i,S_{-i})=\sum_{t}\sum_j\sum_{l}p_{i,j,l,t}D_{j,l,t}\left(\dfrac{1}{2}-\dfrac{p_{i,j,l,t}}{p_{max}}+\dfrac{p_{k,j,l,t}}{2p_{max}}\right)\nonumber\\
& -c_{j,l,t}D_{j,l,t}\left(\dfrac{1}{2}-\dfrac{p_{i,j,l,t}}{p_{max}}+\dfrac{p_{k,l,j,t}}{2p_{max}}\right)-c^{\prime}_{j,l,t}u_{i,j,l,t}.
\end{align}
\end{definition}
$\tilde{F}_i$   depends on the prices selected by the RSP $k$. The optimization problem in $\mathcal{P}$ is equivalent to the following form:
\begin{eqnarray}
\mathcal{P}_1 :  & \text{maximize }  \tilde{F}_i(S_i,S_{-i})\nonumber\\
\text{subject to } & (\ref{eq:state}), (\ref{eq:control}), (\ref{eq:priceconstrain})\nonumber\\
& D_{i,j,l,t}= D_{j,l,t}\left(\dfrac{1}{2}-\dfrac{p_{i,j,l,t}}{p_{max}}+\dfrac{p_{k,l,j,t}}{2p_{max}}\right)\label{eq:dem_pos}\\
& D_{i,j,l,t}\geq 0. \label{eq:equi}
\end{eqnarray}
The difference between this and the previous formulation in $\mathcal{P}$ is that here the maximization with $0$ in (\ref{eq:dem}) is dropped in the objective and instead moved into the constraints.  The objective $\tilde{F}_i$  is concave in the decision variable $p_{i,j,l,t}$; however, it also depends on the prices (strategy) selected by RSP $k$. Further, the set of constraints makes the strategy space $\mathcal S_i(S_{-i})$ a closed compact convex set for any choice of $S_{-i}\in \mathcal{S}_{-i}$.

Because of the constraint $D_{i,j,l,t}\geq 0$, we observe that 
\begin{align}
p_{i,j,l,t}\leq \dfrac{p_{max}}{2}+\dfrac{p_{k,j,l,t}}{2}.
\end{align}
Thus, the solution of the optimization problem gives $p_{i,j,l,t}=\dfrac{p_{max}}{2}+\dfrac{p_{k,j,l,t}}{2}$ when $D_{i,j,l,t}=0$. Hence, it satisfies Assumption~\ref{assum1}. %However, the optimization problem $\mathcal{P}_1$ is now a convex optimization in the strategy of RSP $i$. In other words, given $p_{k,j,l,t}$, the optimization problem is convex. 

%Note that compared to the original demand expression, here, the price $p_{i,j,l,t}$ can be less than $p_{max}$ and yet the demand $\tilde{D}_{i,j,l,t}=0$. However, the two problems $\mathcal{P}$ and $\mathcal{P}_1$ are equivalent. 

%We assume throughout the rest of the paper that the demand and the initial positions of the vehicles are known to all the RSPs. 

\subsection{Potential Game}
Let $\mathcal G_1$ denote the game corresponding to $\mathcal{P}_1$.   With a slight abuse of notation, we will continue to use $\mathcal S_i$ to denote player $i$'s strategy set in this game (as given by the constraints in $\mathcal{P}_1$). We next show that $\mathcal G_1$  is a potential game, i.e., it admits a potential function as defined next. 

\begin{definition}\label{defn:potential}
A function $\Phi(\cdot)$ is a potential for $\mathcal G_1$, if $\Phi(S_i,S_{-i})-\Phi(S^{\prime}_i,S_{-i})=\tilde{F}_i(S_i,S_{-i})-\tilde{F}_i(S_i^{\prime},S_{-i})$ for all $S_i$ and $S^{\prime}_i$ in $\mathcal S_i(S_{-i})$ and all 
$i$.

For differentiable functions, the above definition is equivalent to \cite{monderer1996potential}
\begin{align}
\dfrac{\partial \Phi}{\partial S_i}=\dfrac{\partial \tilde{F}_i}{\partial S_i}\quad \forall i,\nonumber
\end{align}
\end{definition}

Consider the function 
\begin{align}
\Phi=& \sum_{t}\sum_{j}\sum_{l} \dfrac{p_{i,j,l,t}p_{k,j,l,t}D_{j,l,t}}{2p_{max}}+\nonumber\\& \sum_{i}\sum_{t}\sum_{l}\sum_jp_{i,j,l,t}D_{j,l,t}(\dfrac{1}{2}-\dfrac{1}{p_{max}}p_{i,j,l,t})\nonumber\\
& +c_{j,l,t}D_{j,l,t}(\dfrac{1}{p_{max}}p_{i,j,l,t})-c^{\prime}_{j,l,t}u_{i,j,l,t}.
\end{align}

\begin{theorem}\label{thm:pot}
$\mathcal G_1$ is a concave potential game with potential $\Phi$. Solution of the following optimization problem gives a GNE strategy profile under Assumption~\ref{assum1}:-
\begin{align}\label{eq:opt_pot}
\mathcal{P}_{pot}: \text{maximize }\Phi\nonumber\\
\text{subject to } & (\ref{eq:state}),  (\ref{eq:control}), (\ref{eq:dem_pos}), (\ref{eq:equi})\quad  \forall i=1,2.
\end{align}
The decision variables are $p_{i,j,l,t}$, $u_{i,j,l,t}$  for $i=1,2$. 
\end{theorem}

\begin{proof}
Observe that for all $i$,
\begin{align}\label{eq:potential}
\dfrac{\partial \Phi}{\partial p_{i,j,l,t}}=\dfrac{\partial \tilde{F}_i}{\partial p_{i,j,l,t}}\nonumber\\
\dfrac{\partial \Phi}{\partial p_{i,j,l,t}}=\dfrac{\partial \tilde{F}_i}{\partial u_{i,j,l,t}}.
\end{align}
This ensures that $\Phi$ is a potential function of the game by Definition~\ref{defn:potential}. 
 
 Also observe that $\Phi$ has a bilinear term. However, $\Phi$ is still concave because $\dfrac{\partial^2 \Phi}{\partial p_{i,j,l,t}\partial p_{k,j,l,t}}=\dfrac{D_{j,l,t}}{2p_{max}}$ and 
 $\dfrac{\partial^2 \Phi}{\partial p_{i,j,l,t}^2}=-\dfrac{2D_{j,l,t}}{p_{max}}$. Thus, the Hessian is negative definite. 

Further, when we restrict our strategy profile to follow Assumption~\ref{assum1}, the strategy space of each RSP becomes equal to that of the constraints in $\mathcal{P}_{pot}$. Hence, the solution of $\mathcal{P}_{pot}$ is a GNE under Assumption~\ref{assum1}.
%{\bf need to also say something about the strategy sets...}
\end{proof}

\begin{remark}
Since the potential function is strictly concave and the strategy space is closed, compact, and convex, thus, the solution of $\mathcal{P}_{pot}$ is unique. The GNE obtained from the solution of the $\mathcal{P}_{pot}$ is also known as Normalized Nash equilibrium \cite{rosen} or variational-GNE \cite{lygeros}.  
\end{remark}

\begin{remark}
If there is a GNE strategy profile which satisfies Assumption~\ref{assum1}, that will also be the solution of $\mathcal{P}_{pot}$. If GNE is unique and satisfy Assumption~\ref{assum1}, then GNE and the solution of the optimization problem in $\mathcal{P}_{pot}$ are the same. Also note that the solution of $\mathcal{P}_{pot}$ is GNE only under Assumption~\ref{assum1}, i.e., when a RSP $i$ has $0$ demand, its price must be $\dfrac{p_{max}}{2}+\dfrac{p_{k,j,l,t}}{2}$, where $k\in \{1,2\}, k\neq i$. If we relax this assumption, the solution may not be a GNE. 
\end{remark}
\begin{remark}
Note that $\Phi$ is not equal to the sum of the $F_i$'s. Hence, an equilibrium to the potential game is not the same as the outcome obtained when the two RSPs collude  to maximize the sum of their profits. 
\end{remark}

%\begin{remark}
%If for some source-destination pair $(j,l)$ such that in the optimal solution, $D_{i,j,l,t}=0$, but $D_{k,j,l,t}$ is non-zero  then the RSP $k$ has a monopoly over the source destination pair $(j,l)$ at time $t$. 
%\end{remark}

The above result depicts a GNE when each RSP chooses prices under Assumption~\ref{assum1}. In the GNE  obtained by the solution of $\mathcal{P}_{pot}$, the prices for both the RSPs may be different. The next result identifies the condition in which the solution of optimization problem in $\mathcal{P}_{pot}$ gives a symmetric GNE.

\begin{theorem}\label{thm:sym}
If $C_1=C_2$ and all the vehicles are at the same location initially for all the RSPs, then,  if the optimal solution in $\mathcal{P}_{pot}$ is such that for all source-destination pairs $(j,l)$ either all the RSPs have positive demand or zero demand at all time, i.e., either $D_{i,j,l,t}>0, D_{k,j,l,t}>0$ or $D_{i,j,l,t}=D_{k,j,l,t}=0$,  then it must be that $p_{i,j,l,t}=p_{k,j,l,t}$ and  $D_{i,j,l,t}=D_{k,j,l,t}$.
\end{theorem}
The above result indicates that when both the RSPs have the same number of vehicles starting from the same location then there exists a symmetric GNE where the strategies are identical under the condition that no source-destination pair is served by one of the RSPs alone. The result indicates that if all the source-destination pairs are served by the RSPs, in the GNE they can not offer different prices. Further, the total number of customers served is equally divided between the two RSPs. Intuitively, under symmetric conditions where both the RSPs are identical, in the equilibrium both of them should behave similarly. 

%We have already shown that when the demand served by both the RSPs are positive, all the GNEs are identical. Hence, under symmetric conditions and when all the source destination pairs are served by both the RSPs, there is a unique GNE. 

\section{Monopoly Scenario}
We, now, consider the scenario where there is only one (monopoly) RSP. We formulate the monopolist's optimization problem and the strategy of the monopoly. Subsequently, we characterize a condition under which monopolist's strategy profile will be the same as the equilibrium profile in a duopolist setting. 

Since there is only one RSP,  we remove $i\in \{1,2\}$ from the subscript of $p_{i,j,l,t}$ and  simply denote the price from source $j$ to destination $l$ at time $t$ as $p_{j,l,t}$. 

The demand model is now given by the following 
\begin{align}\label{eq:mono_demand}
\hat{D}_{j,l,t}=D_{j,l,t}(1-\dfrac{1}{p_{max}}p_{j,l,t}).
\end{align}
Such linear demand response model has been considered in the literature \cite{arnob_tnse,altman2010}. Note that the demand expression is similar to  (\ref{eq:totdem}) and it is equivalent to (\ref{eq:totdem}) when $p_{i,j,l,t}=p_{k,j,l,t}$ and $D_{i,j,l,t}, D_{k,j,l,t}>0$. Here, we assume that $p_{j,l,t}\leq p_{max}$, thus, $\hat{D}_{j,l,t}\geq 0$. 

We denote the number of free vehicles routed from node $j$ to node $l$ during time slot $t$ as $u_{j,l,t}$. Let the number of available vehicles at location $j$ be $x_{j,t}$. 

Then, similar to (\ref{eq:state}), the number of available vehicles at a location $j$ at time $t$ is given by, 
\begin{align}\label{eq:state_mono}
x_{j,t}=x_{j,t-1}+\sum_{\tau<t}\sum_{l}\hat{D}_{l,j,\tau}A^{\tau}_{l,j,t-\tau}+\nonumber\\
\sum_{\tau<t}\sum_{l}u_{l,j,\tau}A^{\tau}_{l,j,t-\tau}-\sum_{l}\hat{D}_{j,l,t}-\sum_{l}u_{j,l,t}.%\nonumber\\
\end{align}
Here, (\ref{eq:state_mono}) represents the vehicle flow model when there is only one RSP. $x_{i,j,0}$ is the initial location of the vehicles of the monopoly RSP and is assumed to be known. 

The routing of vehicles also must satisfy the following two constraints:
\begin{eqnarray}
 x_{j,t}\geq 0, \quad u_{j,l,t}\geq 0.\label{eq:pos_mono}
\end{eqnarray}
 The constraints in (\ref{eq:pos_mono}) represents the fact that number of available vehicles and the free vehicles that are routed can not be negative. Again, (\ref{eq:state_mono}) and (\ref{eq:pos_mono}) together imply that the total number of vehicles is less than the capacity of the monopoly RSP. % The last expression in (\ref{eq:state_mono}) depicts that the total number of vehicles can not exceed the total number of vehicles $C$ owned by the monopoly RSP. 

%x_{j,t}\geq 0, \quad \sum_{j}x_{j,t}\leq C
Thus, the monopoly RSP solves the following optimization problem:
\begin{eqnarray}
\mathcal{P}_{m}: & \text{maximize } \sum_{j}\sum_{t}p_{j,l,t}D_{j,l,t}(1-\dfrac{1}{p_{max}}p_{j,l,t})\nonumber\\
& -c_{j,l,t}D_{j,l,t}(1-\dfrac{1}{p_{max}}p_{j,l,t})-c^{\prime}_{j,l,t}u_{j,l,t}\label{eq:mono_opti}\\
& \text{subject to } 
(\ref{eq:state_mono}), (\ref{eq:pos_mono}), \quad 0\leq p_{j,l,t}\leq p_{max}\nonumber.
\end{eqnarray}
The above optimization problem is convex. We represent optimal strategy of the RSP as $(p^{*}_{j,l,t},u^{*}_{j,l,t})$. We can directly compare the solution of the monopoly setting with the duopoly setting for example from the solution of the potential game $\mathcal{P}_{pot}$ (cf.~(\ref{eq:opt_pot})). %Note that though in duopoly there can be other GNEs, the one obtained from the solution of $\mathcal{P}_{pot}$ is the computationally efficient one. Hence, it is reasonable to assume that the RSPs will implement the equilibrium given by $\mathcal{P}_{pot}$. 
%Here, we represent the problem as a convex optimization problem by representing the demand as in (\ref{eq:mono_demand}). 

\subsection{Equivalence between Monopoly and Duopoly}  
We now describe a condition under which the monopoly price strategy and duopoly price strategy coincide at all the origin-destination pairs for all time. Thus, competition will not have any impact on the pricing strategy.  Note that a social planner or regulator may want to maximize the users' surplus which increases as the prices decrease. Thus, the social planner or regulator may want to regulate the market in avoid conditions where competition does not impact the price. 

Suppose that the total vehicles are partitioned into two sets $\mathcal{N}_1$ and $\mathcal{N}_2$. % Also assume that the vehicles are initially at the same position. 
%Suppose that the total vehicles are partitioned in two sets $\mathcal{N}_1$ and $\mathcal{N}_2$. Also assume that the vehicles are initially at the same position.  
%
Let $x^{i}_{j,t}$ be the available vehicles at location $j$ at time $t$ that belong to $\mathcal{N}_i$, for $i=1,2$. Further, let $\hat{D}^{i}_{j,l,t}$ be the demand which will be satisfied by the vehicles in set $\mathcal{N}_i$ and let $u^{i}_{j,l,t}$ be the number of free vehicles among the $x^{i}_{j,t}$ vehicles which are routed from location $j$ to $l$. Then, $x^{i}_{j,t}$  changes as
\begin{align}\label{eq:state_part}
x^{i}_{j,t}=x^{i}_{j,t-1}+\sum_{\tau<t}\sum_{l}\hat{D}^{i}_{l,j,\tau}A^{\tau}_{l,j,t-\tau}+\nonumber\\
\sum_{\tau<t}\sum_{l}u^{i}_{l,j,\tau}A^{\tau}_{l,j,t-\tau}-\sum_{l}\hat{D}^{i}_{j,l,t}-\sum_{l}u^{i}_{j,l,t}
\end{align}
where
\begin{align}
\sum_{j}x^{i}_{j,t}\leq |\mathcal{N}_i|
\end{align}
and $|\mathcal{N}_i|$ is the cardinality of $\mathcal{N}_i$.  Here, $x^{i}_{j,0}$ is the initial locations of the vehicles in the set $\mathcal{N}_i$ at location $j$. 

Now, consider the following optimization problem 
\begin{align}\label{eq:opti_part}
\mathcal{P}_{m,1}:\text{maximize } & \sum_{j}\sum_{t}p_{j,l,t}D_{j,l,t}\left(1-\dfrac{1}{p_{max}}p_{j,l,t}\right)\nonumber\\
& -c_{j,l,t}D_{j,l,t}\left(1-\dfrac{1}{p_{max}}p_{j,l,t}\right)-\nonumber\\ & c^{\prime}_{j,l,t}(u^{1}_{j,l,t}+u^{2}_{j,l,t})\nonumber\\
\text{subject to } 
& \hat{D}^{1}_{j,l,t}+\hat{D}^{2}_{j,l,t}=D_{j,l,t}\left(1-\dfrac{1}{p_{max}}p_{j,l,t}\right)\nonumber\\
& \hat{D}^{i}_{j,l,t}\geq 0\quad  \forall i\nonumber\\
& (\ref{eq:state_part}), 0\leq p_{j,l,t}\leq p_{max},\quad u^{i}_{j,l,t}\geq 0.\nonumber
\end{align}
This is essentially the same as the monopolist optimization problem presented earlier in $\mathcal{P}_m$, except here we track which vehicle is in each of the two partitions. The next result characterizes when the monopoly and the duopoly price will be the same. We denote the optimal solution to $\mathcal{P}_{m,1}$ as $(p^{*}_{j,l,t},u^{*,1}_{j,l,t},u^{*,2}_{j,l,t})$. Likewise, $u^{*,1}_{i,j,l,t}+u^{*,2}_{j,l,t}$ gives the optimal solution $u^{*}_{j,l,t}$ of the optimization problem $\mathcal{P}_m$ (cf. (\ref{eq:mono_opti})).

%The optimal solution of (\ref{eq:opti_part}) is equal to the optimal solution of the optimal solution of $\mathcal{P}_m$ in (\ref{eq:mono_opti}) with $p^{*}_{j,l,t}=p^{*}_{j,l,t}$, and $

\begin{theorem}
Suppose that the solution of the optimization problem in $\mathcal{P}_{m,1}$ is such that $\hat{D}^{i}_{j,l,t}\hat{D}^{k}_{j,l,t}=0$ $\forall (j,l,t)$. Then, the optimal solution is also a GNE profile of the duopoly scenario when $x^{i}_{j,0}=x_{i,j,0}$, and $|\mathcal{N}_i|=C_i$, for $i=1,2$ with the GNE strategy profile $p_{i,j,l,t}=p_{j,l,t}$ if $\hat{D}^{i}_{j,l,t}>0$, $p_{i,j,l,t}=p_{max}$ if $\hat{D}^{i}_{j,l,t}=0$, $u^{i}_{j,l,t}=u_{i,j,l,t}$ for $i=1,2$.

Under the above GNE strategy profile, $D_{i,j,l,t}=\hat{D}^{i}_{j,l,t}$ for $i=1,2$.
\end{theorem}
Note that in order to have the monopoly price equal to the duopoly price in a GNE, the price must be set at $p_{i,j,l,t}=p_{max}$ when $D_{i,j,l,t}=0$ which violates Assumption~\ref{assum1}. Hence, the above solution will not be equal to the GNE given by the solution of $\mathcal{P}_{pot}$ (cf.~(\ref{eq:opt_pot})). 

%Thus, under Assumption 1 we can not have monopoly price equal to the duopoly price. 
Intuitively, when $p_{i,j,l,t}=p_{max}$, the demand for RSP $k$, $k\neq i, k\in \{1,2\}$ becomes equal to the expression of $\hat{D}_{j,l,t}$ with $p_{k,j,l,t}=p_{j,l,t}$. Hence,  if an optimal solution in $\mathcal{P}_{m,1}$ is such that $\hat{D}^{i}_{j,l,t}\hat{D}^{k}_{j,l,t}=0$, then in the duopoly setting we have a GNE which is same as the monopoly solution. 

Note that when $\hat{D}^{i}_{j,l,t}\hat{D}^{k}_{j,l,t}=0$ for every source destination pair $(j,l)$, the monopoly RSP also partitions the total demand into two mutually exclusive sets where each set of demand is served by only one set of vehicles. Thus, an optimal solution for the monopoly is such that each set of demand can be served by the set of vehicles $\mathcal{N}_i$ $i=1,2$ with the same cardinality as the capacity of RSP $i$ and same initial locations as of the vehicles of RSP $i$. It follows that  the optimal solution also induces a GNE in the duopoly scenario such that in the GNE strategy profile  $p_{i,j,l,t}=p_{max}$ when $D_{i,j,l,t}=0$. 

Note that if the condition $\hat{D}^{i}_{j,l,t}\hat{D}^{k}_{j,l,t}=0$  for all pairs $(j,l)$ is not satisfied in any optimal solution of $\mathcal{P}_{m,1}$, then it {\em may not} induce the GNE even when in the GNE strategy profile $p_{i,j,l,t}=p_{max}$ when $D_{i,j,l,t}=0$.  

%\begin{color}{red}
\section{Extension: Non-Deterministic Demand}\label{sec:stoc_demand}
Though RSPs in general estimate the demand between any two locations, our analysis will go through when the RSPs only know the distribution of the demand between any two nodes $j$ and $l$, $D_{j,l,t}$. %

A scenario denotes the joint demand across the time-periods $t=1,\ldots, T$. Suppose that the scenario $m$ where the joint demand for $t=1,\ldots,T$, $D_{j,l,t}=D_{j,l,t}^m$ occurs with probability (w.p.) $q_m$. Let $D_{i,j,l,t}^m$ be the demand from location $j$ to $l$ for RSP $i$ for the $m$-th scenario which is obtained using (\ref{eq:dem}) where $D_{j,l,t}^m$ replaces $D_{j,l,t}$.  We can formulate a scenario-based stochastic programming problem for each RSP where each RSP wants to maximize its expected profit which we describe in the following. 

 The objective function for each RSP is now $F^{\prime}_i=\sum_{t}\sum_{j}\sum_{l}(\sum_{m}q_mp_{i,j,l,t}D_{i,j,,l,t}^m-c_{j,l,t}D_{i,j,l,t}^m)-c^{\prime}_{j,l,t}u_{i,j,l,t}$). The constraint in (\ref{eq:state})  is replaced with the following--
\begin{eqnarray}
x_{i,j,t}^m=x_{i,j,t-1}^m+\sum_{\tau<t}\sum_{l}D^m_{i,l,j,\tau}A^{\tau}_{l,j,t-\tau}+\nonumber\\
\sum_{\tau<t}\sum_{l}u_{i,l,j,\tau}A^{\tau}_{l,j,t-\tau}-\sum_{l}D^m_{i,j,l,t}-\sum_{l}u_{i,j,l,t}\quad \forall m.\label{eq:state_stoc}
\end{eqnarray}
$x_{i,j,t}^m$ is the number of vehicles of RSP $i$ at the location $j$ for the $m$-th scenario. $x_{i,j,0}^m$ is the initial location of vehicles which is the same for all the scenarios. 

The constraint in (\ref{eq:control}) is replaced by
\begin{equation}\label{eq:control_stoc}
x_{i,j,t}^m\geq 0, \quad u_{i,j,l,t}\geq 0.
\end{equation}
Hence, the modified optimization problem for RSP $i$ is now
\begin{align}
\mathcal{P}_{stoc}: & \text{ maximize } F^{\prime}_i\nonumber\\
& \text{subject to } (\ref{eq:state_stoc}), (\ref{eq:control_stoc}), (\ref{eq:priceconstrain}).\nonumber
\end{align}
Each RSP $i$ takes its decision such that the constraints for all the scenarios are satisfied which can be quite restrictive. One can also consider an approach where constraints for at least $M_1$ number of scenarios are satisfied where the probability that at least one of the $M_1$ scenarios occurs is large. Since the constraints are affine and the structure of the demand function remains the same, one can define a potential game by representing the demand in an equivalent form as in (\ref{eq:dem_pos}) and (\ref{eq:equi}) for each scenario and then restricting the problem for Assumption~\ref{assum1}. The detailed analysis is omitted here. 

%\end{color}
%\begin{align}

%\end{align}
%We will assume that the RSPs have the distribution for the demand and the RSPs want to maximize the 
\section{Numerical Simulations}\label{sec:numerical}
\subsection{Set Up}
In this section, we numerically evaluate the pricing strategy when there are two RSPs and compare with the monopoly price. Towards this end, we consider a network consisting of two cluster of nodes (Fig.~\ref{fig:cluster}). Each cluster consists of $n=10$ nodes. Nodes in each cluster are connected with each other. The time unit taken from any node of the cluster to any other node within the cluster is $1$. There also exists an edge from each node of the cluster to every other node of the other cluster. However, the time required to travel between clusters is $2$. 

\begin{figure}
\includegraphics[width=80mm, height=30mm]{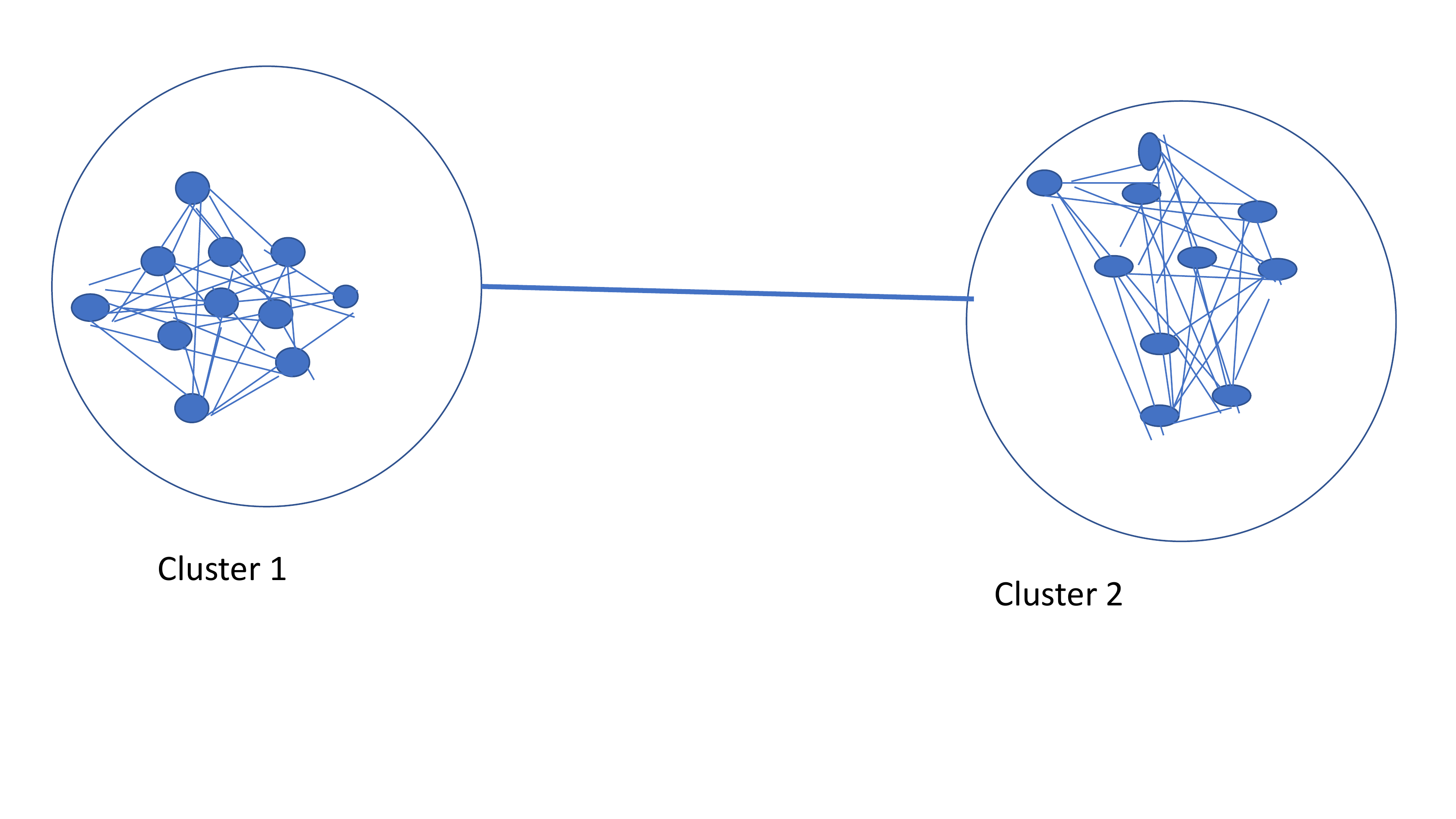}
\vspace{-0.2in}
\caption{Two clusters: each cluster has $10$ nodes. Every node is connected to every other node in the graph. However, the time taken to go from one node within a cluster to a node of different cluster is $2$ units compared to time taken to go between any two nodes within a cluster which is $1$ unit.}
\label{fig:cluster}
\vspace{-0.2in}
\end{figure}
We consider $T=4$. Total outgoing demand from each node is the same, $D_{m,t}$ at time $t$. We consider that the demand towards every node within the cluster as $(1-q)D_{m,t}/(n-1)$ and the demand towards each node of the other cluster is $qD_{m,t}/n$ where $q\in (0,0.5)$. %{\bf I changed the demand within a cluster to be divided by $n-1$ since I assume that there is not demand from a node back to itself - correct?}
Note that when $q=0.5$, the demand becomes balanced as the demand towards any node from a given node is the same. When $q$ is near $0$, there is little outgoing demand towards the nodes of the other cluster. Note that such a demand can be prevalent in a city consisting of two highly populated suburban areas, and the demand within each area can be considered to be uniformly distributed where as the demand from one sub-urban area to another may vary over time.

We assume that demand is $D_{m,1}=40, D_{m,2}=20, D_{m,3}=40, D_{m,4}=40$. Thus, the demand oscillates between a high value and a low value. We assume that each RSP has the same capacity $C_i$. Similar to the demand, we consider that  at each node in the cluster $i$, initially RSP $i$ has $(1-q)C_i/n$  vehicles and at every node of the other cluster the number of vehicles is $qC_i/n$. Again as $q$ increases, the initial positions of the vehicles become more uniformly distributed over all the locations. The cost $c_{j,l,t}=0.1$ ($c^{\prime}_{j,l,t}=0.05$) for all origin-destination pairs within the cluster and $c_{j,l,t}=0.2$ ($c^{\prime}_{j,l,t}=0.1$) for all origin-destination pairs which originate from one node in a cluster and ends in a different cluster. We consider $p_{max}=1$. We compare the prices with the monopoly setting where we assume that a single RSP serves the customers with combined capacity of the RSPs (i.e, $C_1+C_2$) and with the same initial locations as that of the vehicles of RSP $1$ and RSP $2$. 

%\begin{figure*}
%\begin{minipage}{0.24\textwidth}
%\includegraphics[width=0.96\linewidth]{price_1.eps}
%\caption{Variation of Prices at time periods $2$, and $4$ (low demand regime) between any two nodes within cluster $1$ with $q$. $p_i$ denotes the prices of RSP $i$, $i=1,2$. $p_m$ denotes the monopoly price. }
%\label{fig:q}
%\end{minipage}\hfill
%\begin{minipage}{0.24\textwidth}
%\includegraphics[width=0.96\linewidth]{price_high.eps}
%\caption{Variation of Prices at time periods $1$ and $3$ (high demand regime) between any two nodes within cluster $1$ with $q$. $p_i$ denotes the prices of RSP $i$, $i=1,2$. $p_m$ denotes the monopoly price. }
%\label{fig:q_diff}
%\end{minipage}\hfill
%\begin{minipage}{0.24\textwidth}
%\includegraphics[width=0.96\linewidth]{price_diff.eps}
%\caption{Variation of Prices at time periods $1$ and $3$ (low demand regime) between any node from one cluster in cluster $1$ to a node at cluster $2$ with $q$. $p_i$ denotes the prices of RSP $i$, $i=1,2$. $p_m$ denotes the monopoly price.}
%\label{fig:high}
%\end{minipage}\hfill
%\begin{minipage}{0.24\textwidth}
%\includegraphics[width=0.96\linewidth]{price_highc.eps}
%\caption{Variation of Prices at time period $4$ (high demand regime) between any two nodes within cluster $1$ with $q$ and with $C_i=800$. $p_i$ denotes the prices of RSP $i$, $i=1,2$. $p_m$ denotes the monopoly price.}
%\label{fig:highc}
%\end{minipage}
%\end{figure*}

\begin{figure}
\includegraphics[width=0.7\linewidth]{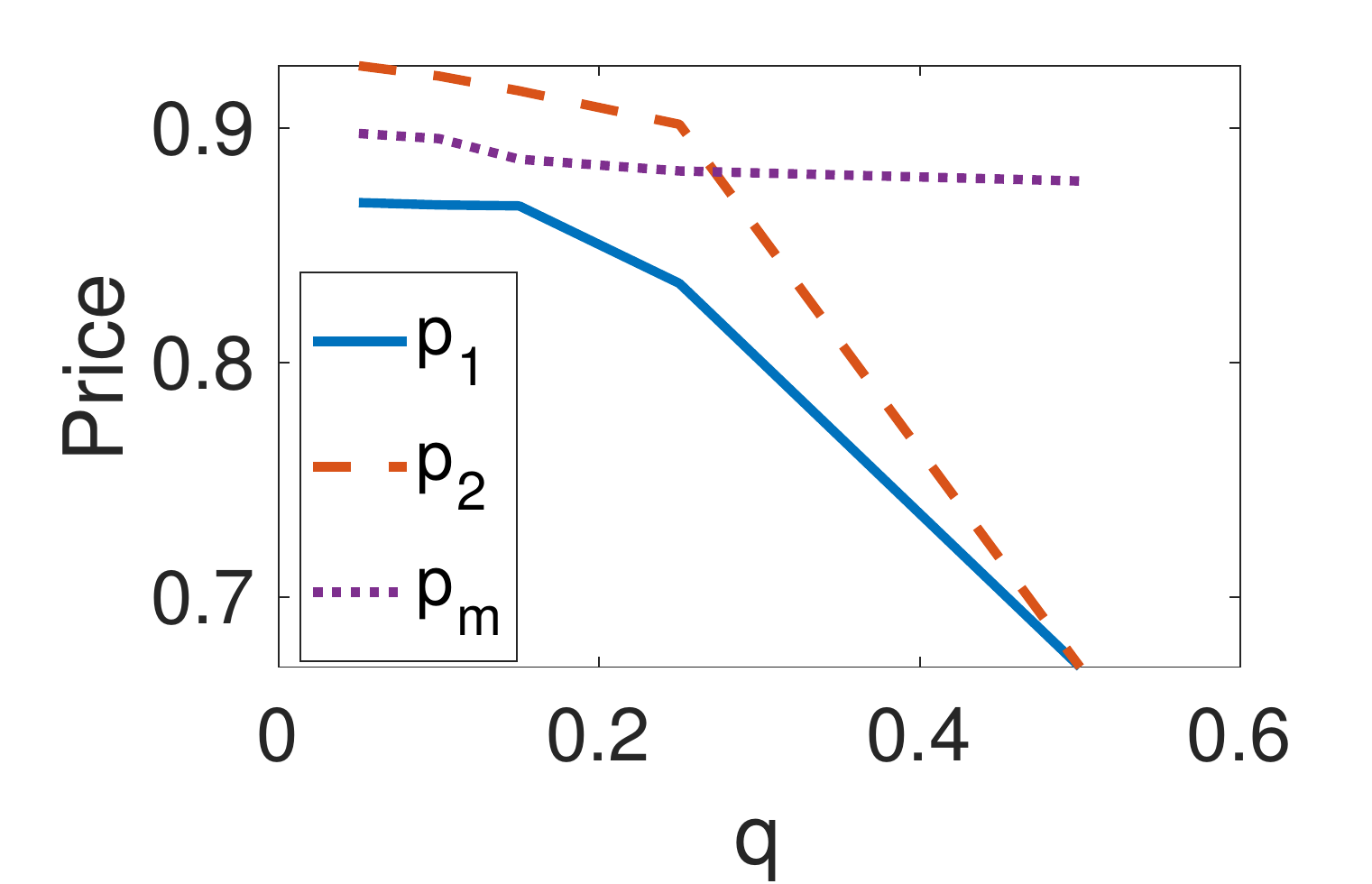}
\vspace{-0.1in}
\caption{Variation of Prices at time periods $2$, and $4$ (low demand regime) between any two nodes within cluster $1$ with $q$. $p_i$ denotes the prices of RSP $i$, $i=1,2$. $p_m$ denotes the monopoly price. }
\label{fig:q}
\vspace{-0.2in}
\end{figure}

\begin{figure}
\includegraphics[width=0.7\linewidth]{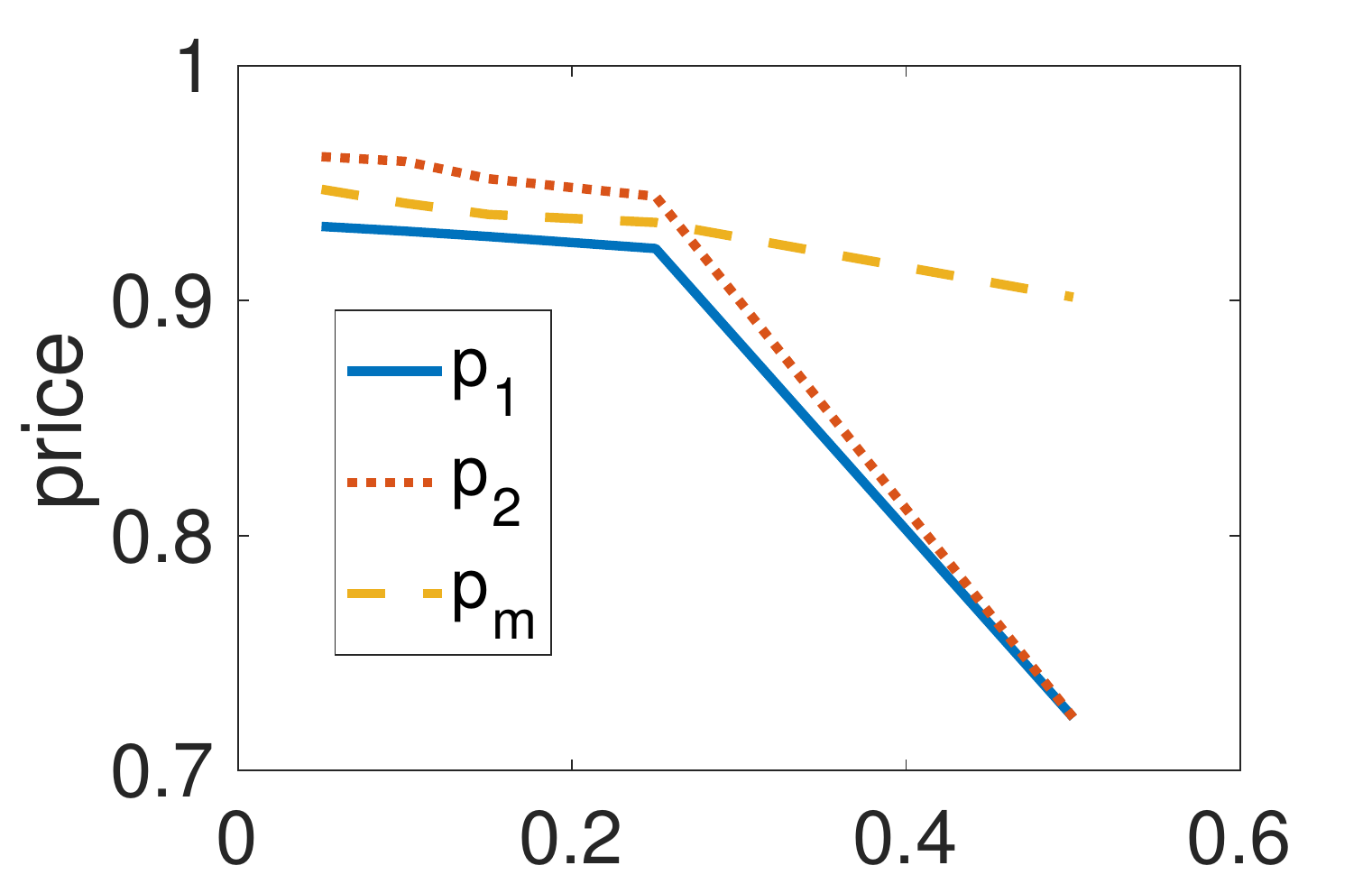}
\vspace{-0.1in}
\caption{Variation of Prices at time periods $1$ and $3$ (high demand regime) between any two nodes within cluster $1$ with $q$. $p_i$ denotes the prices of RSP $i$, $i=1,2$. $p_m$ denotes the monopoly price. }
\label{fig:q_diff}
\vspace{-0.2in}
\end{figure}

\begin{figure}
\includegraphics[width=0.6\linewidth]{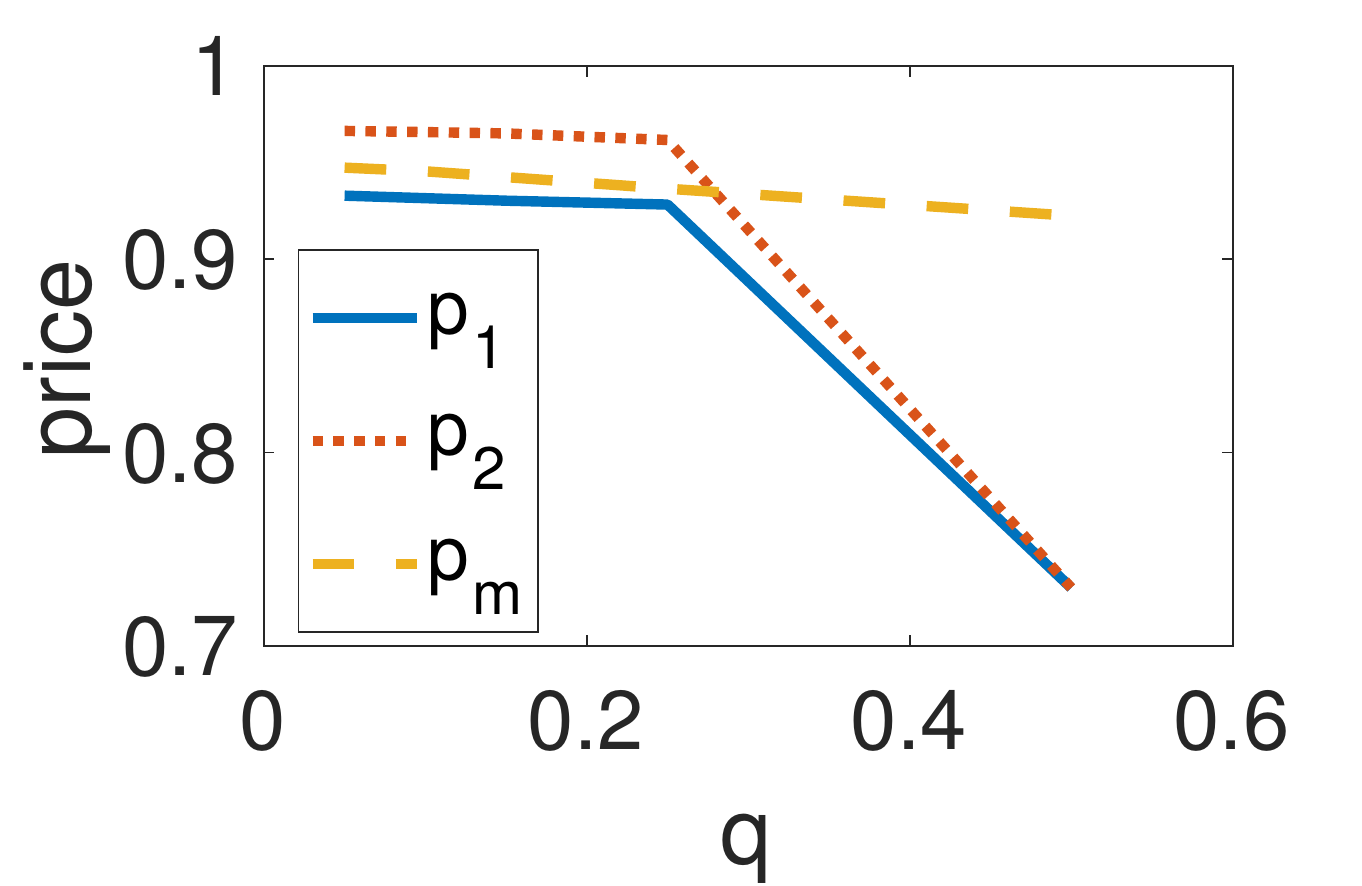}
\vspace{-0.1in}
\caption{Variation of Prices at time periods $1$ and $3$ (low demand regime) between any node from one cluster in cluster $1$ to a node at cluster $2$ with $q$. $p_i$ denotes the prices of RSP $i$, $i=1,2$. $p_m$ denotes the monopoly price.}
\label{fig:high}
\vspace{-0.2in}
\end{figure}

\begin{figure}
\includegraphics[width=0.6\linewidth]{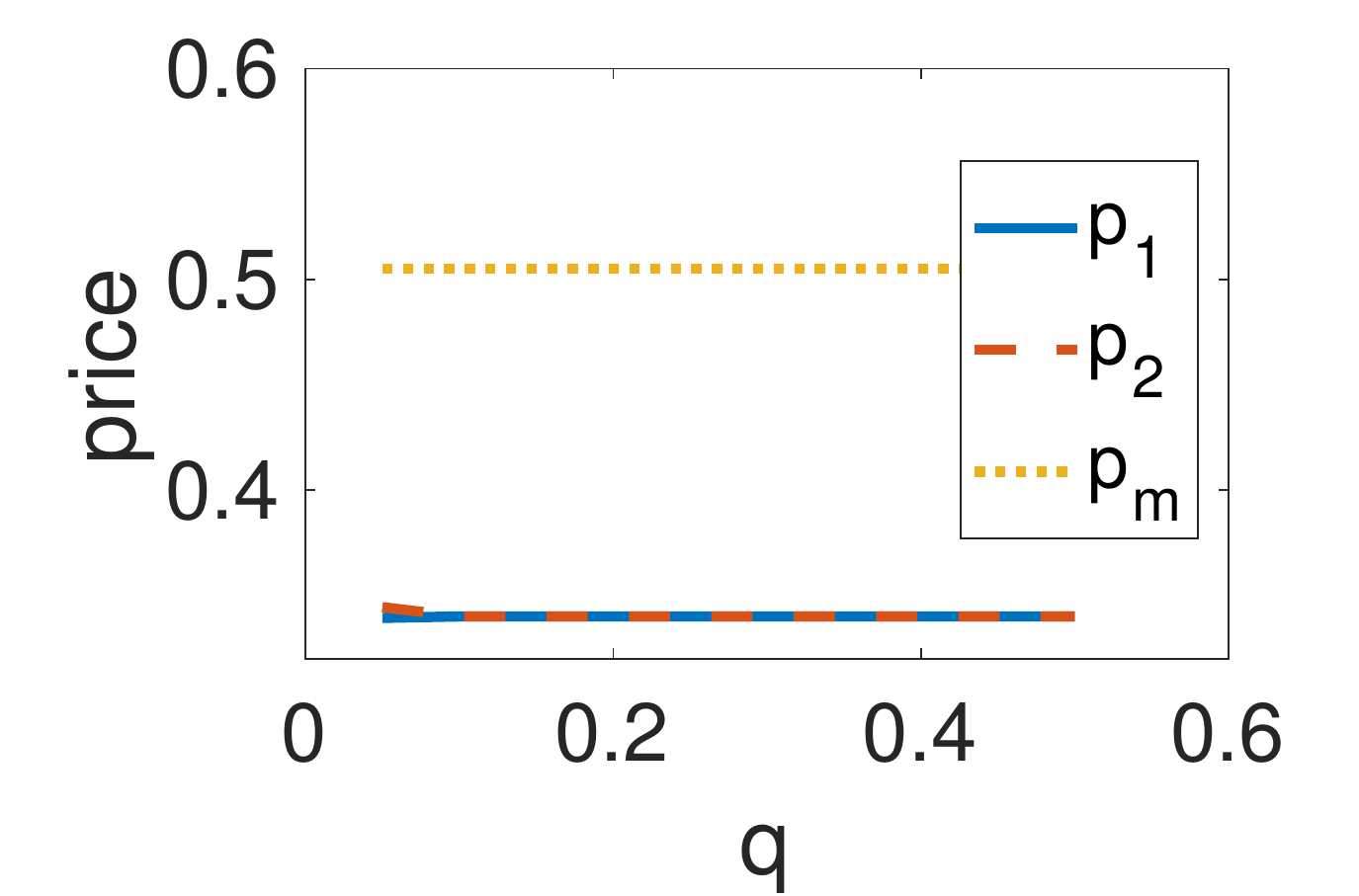}
\vspace{-0.1in}
\caption{Variation of Prices at time period $4$ (high demand regime) between any two nodes within cluster $1$ with $q$ and with $C_i=800$. $p_i$ denotes the prices of RSP $i$, $i=1,2$. $p_m$ denotes the monopoly price.}
\label{fig:highc}
\vspace{-0.2in}
\end{figure}

\subsection{Results}
\subsubsection{Impact of $q$}
We compare the prices  for any origin-destination pair within the same cluster. In Fig.~\ref{fig:q}, we show the price  for demand which originates and ends within cluster $1$. Recall that when $q$ is small, the demand is small from one cluster to another cluster. Also note that when $q$ is small,  RSP $1$ has  most of its initial vehicles stationed at the nodes within cluster $1$. Thus,  RSP $1$ chooses very high price for every demand within cluster $1$ since it faces little competition from RSP $2$. Hence, RSP $2$ and $1$ divide the entire region with RSP $1$ serving  cluster $1$ and RSP $2$ serving cluster $2$. Note that when $q$ is small, the prices almost become equal to the monopoly prices. 

Note that when $q$ increases, the prices of both the RSPs decrease even for the origin-destination pairs which originate and end within the same cluster. Thus, the competitive price becomes different from the monopoly price. When $q$ increases, the  demand becomes more balanced and the initial locations of the vehicles of the RSPs also become more balanced over the network. Hence, competition increases. When $q$ becomes $0.5$, the maximum competition is reached where the prices for both the RSPs become the same which is in accordance of our theoretical findings in Theorem~\ref{thm:sym}. 

Fig.~\ref{fig:q_diff} shows the prices when a demand originates from node within a cluster and ends at a node in the other cluster. The price is higher compared to the price for the demand originating and ending within the same cluster. Since, here, a RSP utilizes fewer vehicles to serve the demand as it takes more time to reach the destination which may result in a potential loss of revenue. Similar to Fig.~\ref{fig:q}, here, as $q$ increases the prices  of the RSPs decrease sharply and become equal when $q=0.5$. 

\subsubsection{Impact of High Demand}
Since the demand becomes higher during the first and third time period, the corresponding prices are also higher. Fig.~\ref{fig:high} depicts how  RSP $1$ selects prices for the origin-destination pairs within the same cluster. The price of RSP $1$ is smaller compared to RSP $2$ similar to Fig.~\ref{fig:q}. The prices of both the RSPs decrease as $q$ increases. 

Note that in the high demand period, the price from any demand originating from one node in a cluster to a node in other cluster reaches the maximum price $p_m$. Hence, the RSPs only compete for customers who want to travel only within its own cluster. This shows that in a high demand regime, the RSPs may not even serve some demand from one location to a distant location even if the demand is uniformly distributed. 

\subsubsection{Impact of High capacity}
In Fig.~\ref{fig:highc}, we consider the scenario where the RSPs have a larger number of vehicles. Specifically, we consider that $C_1=C_2=800$. The figure shows that since the number of vehicles becomes very large, the RSPs would compete more fiercely for nearly all values of $q$. Hence, the prices become almost independent of $q$, and the prices are much lower compared to the monopoly price. 

\section{Conclusions}

We considered a model for competition among RSPs that accounts for their pricing and routing decisions over time. Under an assumption about how vehicles price when demand is zero, we showed that this game admits a potential function and used this to characterize a GNE of the game.  We also compared a market with a single monopoly RSP to one with two competitive RSPs and showed that the impact of competition may depend on the spatial distribution of demand and the size of the RSPs.

\bibliographystyle{IEEEtran}
\bibliography{rsp_competition}
\end{document}